\documentclass[reqno]{amsart}

\usepackage{tikz}
\usetikzlibrary{arrows.meta}

\usepackage{a4wide}
\usepackage{amsmath,amssymb,amsfonts,amsthm}
\usepackage{moreverb,rotating,graphics}
\usepackage[utf8]{inputenc}
\usepackage[english]{babel} 
\usepackage{framed,fancybox}
\usepackage{caption, subcaption}
\usepackage{mathrsfs, esint, dsfont}
\usepackage{eurosym}
\usepackage[colorlinks, linkcolor={blue}, citecolor={red}, urlcolor = {blue}]{hyperref}
\usepackage{cancel}
\usepackage{bbm}
\usepackage{caption}
\usepackage{float}
\usepackage{subcaption}
\usepackage{enumerate}
\usepackage{mathrsfs}  
\usepackage{textcomp}
\usepackage{color}
\usepackage{bbm}
\usepackage{graphicx}
\usepackage{enumitem}
\usepackage{enumerate}
\usepackage{listings}
\usepackage{fancyhdr}
\usepackage{csquotes}
\usepackage{hyperref}

\newtheorem{theorem}{Theorem}[section]
\newtheorem{proposition}[theorem]{Proposition}
\newtheorem{remark}[theorem]{Remark}
\newtheorem{lemma}[theorem]{Lemma}

\newtheorem{definition}[theorem]{Definition}

\newcommand{\ri}{\mathrm{i}}

\newcommand{\cE}{\mathcal{E}}

\newcommand{\C}{\mathbb{C}}
\newcommand{\R}{\mathbb{R}}
\newcommand{\N}{\mathbb{N}}
\newcommand{\Z}{\mathbb{Z}}

\newcommand\ba{\mathbf{a}}

\newcommand\bh{\mathbf{h}}

\newcommand\bk{\mathbf{k}}

\newcommand\bQ{\mathbf{Q}}

\newcommand\bt{\mathbf{t}}
\newcommand\bu{\mathbf{u}}

\newcommand\bv{\mathbf{v}}

\newcommand\bw{\mathbf{w}}

\newcommand\bnull{\mathbf{0}}

\def\cB{{\mathcal B}}

\def\cE{{\mathcal E}}
\def\cF{{\mathcal F}}

\def\fS{{\mathfrak S}}

\def\sC{{\mathscr{C}}}

\def\sL{{\mathscr{L}}}

\def\rd{{\mathrm{d}}}
\def\re{{\mathrm{e}}}
\def\ri{{\mathrm{i}}}

\newcommand\1{{\ensuremath {\mathds 1} }}

\def\bra{\langle}
\def\ket{\rangle}

\def\Tr{{\rm Tr} }

\def\Ker{{\rm Ker }}

\def\bnull{{\mathbf 0}}

\def\G{\Gamma}

\newcommand{\msc}[1]{\href{https://mathscinet.ams.org/mathscinet/search/mscdoc.html?code=#1}{#1}}

\newcommand{\myfootnote}[1]{
    \renewcommand{\thefootnote}{}
    \footnotetext{\scriptsize#1}
    \renewcommand{\thefootnote}{\arabic{footnote}}
}

\title[Exponential decay of critical points]{Exponential decay of the critical points in a discrete model of polyacetylene}

\author{David Gontier \qquad Adechola E. K. Kouande \qquad Éric Séré}

\date{\today}

\begin{document}
    \myfootnote{David Gontier: CEREMADE, Université Paris-Dauphine, PSL University,75016 Paris, France \& ENS/PSL University, DMA, F-75005, Paris, France;\\
        email: \href{gontier@ceremade.dauphine.fr}{gontier@ceremade.dauphine.fr}}
    \myfootnote{Adéchola E. K. Kouande: CEREMADE, University of Paris-Dauphine, PSL University, 75016 Paris, France\\
        email: \href{kouande@ceremade.dauphine.fr}{kouande@ceremade.dauphine.fr}}
    \myfootnote{Éric Séré: CEREMADE, University of Paris-Dauphine, PSL University, 75016 Paris, France\\
        email: \href{sere@ceremade.dauphine.fr}{sere@ceremade.dauphine.fr}}
    
    \begin{abstract}
        In this paper we consider stationary states of the SSH model for infinite polyacetylene chains that are homoclinic or heteroclinic connections between two-periodic dimerized states. We prove that such connections converge exponentially fast to the corresponding asymptotic periodic states. 
        
        \bigskip
        \noindent \sl \copyright~2024 by the authors. This paper may be reproduced, in its entirety, for non-commercial purposes.
    \end{abstract}


    \maketitle    
    
    \tableofcontents


    \section{Introduction}
    
    The goal of this article is to prove an exponential decay property for critical points in the SSH energy. The SSH model was introduced by Su, Schrieffer and Heeger to describe polyacetylene, which is a long one-dimensional chain of carbon (and hydrogen) atoms. In this model, the energy may be written as a function of a real sequence ${\bf t}=(t_n)_{n\in I}$. Each number $t_n$ represents an electronic hopping amplitude between the $n$-th atom and the $(n+1)$-th atom. It is related to their distance $d_n$ by an affine relation of the form $t_n=a - b d_n$. The index set $I$ is $\mathbb{Z}/L\mathbb{Z}$ for a closed chain with $L$ carbon atoms, and $I=\mathbb{Z}$ for an infinite chain. It is a well known fact that the chain can lower its total energy by dimerizing. This physical phenomenon was first predicted by Peierls~\cite{Pei96} (see also~\cite{Fro54}) and is now known as the Peierls distortion or dimerization.
    
    Actually, Kennedy and Lieb~\cite{KenLie04}, and Lieb and Nachtergaele~\cite{LieNac95} proved that the minimizers of the SSH energy associated to closed polyacetylene with an even number $L$ of carbon atoms at zero temperature are always $2$-periodic. When $L \equiv 2 \ {\rm mod} \ 4$ or when $L \equiv 0 \ {\rm mod} \ 4$ is large enough, these minimizers are dimerized, in the sense that they are $2$-periodic but not $1$ periodic. In this situation, there are exactly two minimizing configurations of the energy, that we call $\bt^-$ and $\bt^+$. They are of the form
    \begin{equation}\label{dimerized_config}
        t^-_n := W - (-1)^n \delta, \quad \text{and} \quad 
        t^+_n := W + (-1)^n \delta.
    \end{equation}
    The quantity $W > 0$ is the average hopping amplitude between neighbouring carbon atoms, and $\delta > 0$ is the distortion. At zero temperature this distortion phenomenon remains in the thermodynamic limit, that is when the number of carbon atoms $L$ goes to infinity. However, we proved in \cite{GonKouSer23} that the distortion disappears at a sufficiently high temperature. More precisely there is a critical temperature below which the polyacetylene chain is dimerized ($\delta > 0$) and behaves as an insulator, and above which it is 1-periodic ($\delta = 0$) and behaves as a metal.

    \medskip
    
    In this paper, we study a family of critical points, namely homoclinic and heteroclinic critical points, for the infinite SSH model at zero temperature.
    Using the terminology of dynamical systems, a configuration $\bt = (t_n)_{n \in \Z}$ is said to be {\bf homoclinic} if
    \begin{equation}\label{homoclinic_conf}
        \lim_{n \to - \infty} | t_n - t_n^+ | =  
        \lim_{n \to + \infty} | t_n - t_n^+ | = 0,
    \end{equation}
    and it is said to be {\bf heteroclinic} if
    \begin{equation}\label{heteroclinic_conf}
        \lim_{n \to - \infty} | t_n - t_n^- | =  
        \lim_{n \to + \infty} | t_n - t_n^+ | = 0.
    \end{equation}
    In other words, a configuration is homoclinic if it converges to the same dimerized state $\bt^+$ at infinity, and it is heteroclinic if it switches from the dimerized state $\bt^-$ at $-\infty$ to the dimerized state $\bt^+$ at $+ \infty$. Of course, one can also consider homoclinic and heteroclinic configurations which converge to $\bt^-$ at $+ \infty$. But they can be obtained from the other ones by the shift $n\mapsto n+1$.
    
    \medskip
    
    The expression ``critical points'' refers to solutions of the Euler-Lagrange equation associated with the SSH energy functional. The fact that heteroclinic critical points exist was suggested in~\cite{SuSchHee79, SuSchHee80, KivHei82, Su80, VosAalSaa96}. A rigorous proof was given in \cite{Gar11, GarSer12} and such configurations are sometimes called {\em kinks}.
    
    \medskip
    
    Let us mention that the discrete model has a continuous analogue, first derived by Takayama, Lin-Liu and Maki~\cite{TakLinMak80} in the context of polyacetylene. It turns out that this continuous model coincides with the so-called Gross--Neveu model~\cite{GroNev74}, which has explicit kink critical points as well as explicit homoclinic critical points, that have been completely classified~\cite{Fei95, Fei04}. These explicit solutions converge exponentially fast to the stable states when the one-dimensional space variable converges to $\pm \infty$.

    \medskip
    
  Our main goal is to prove that homoclinic and heteroclinic critical points of the discrete infinite SSH model converge exponentially fast to the corresponding dimerized states. We do not prove that this convergence rate is optimal, but we are convinced that it is, as we expect a behaviour similar to the continuum model. Numerical simulations also seem to confirm this.
  
  \medskip
  
  To the best of our knowledge, there is no rigorous result relating the discrete and continuum models. It would be interesting to prove that critical points of the continuum model approximate critical points of the discrete model in a certain regime of parameters. It would also be interesting to study the interaction between kinks in the discrete model. We see the present work as a preliminary step in these directions.

    
    \section{Critical points for the infinite SSH model, and main result}
    \label{sec:critical_point_contour}
    
    \subsection{The SSH model}
    \label{ssec:Peierls}
    
    In order to state our theorem, we first need to define the notion of critical points for our infinite SSH model. Let us start with the SSH model for a periodic linear chain with $L$ classical atoms, together with quantum non-interacting electrons. We will take the limit $L \to \infty$ at the end. 
    
    We denote by $t_n$ the hopping amplitude between the $n$-th and $(n + 1)$-th atoms (modulo $L$), and set $\bt:=~( t_n)_{n\in\mathbb{Z}/L\mathbb{Z}}$. In what follows, we only consider the case where all hopping amplitudes are strictly positive, and we sometimes write $\bt \ge \tau > 0$ as a shortcut for $t_n \ge \tau$ for all $n$. Note that the dimerized configurations automatically satisfy this condition. Indeed, in \cite{KenLie04}, it is proved that $\delta < W$, so $t_n \geq \tau$, with $\tau =  W -\delta > 0$.
    
    \medskip
    
    The SSH energy of the closed chain of length $L$ is (see~\cite{KenLie04, LieSchMat61, MacNac96,Pei96, SuSchHee79, GonKouSer23})
    \begin{equation} \label{eq:def:Peierls_finite}
        \cE^{(L)}( \bt ) := \frac{\mu}{2} \sum_{n\in\mathbb{Z}/L\mathbb{Z}}^{}(t_n - 1)^2 - 2 \Tr( T_- ),
    \end{equation}
    where $T = T( \bt )$ is the $L \times L$ hermitian matrix
    \begin{equation}\label{CH1_3}
        T = T( \bt ) := \begin{pmatrix}
            0      & t_1    & 0      & 0         & \cdots   & t_{L} \\
            t_1    & 0      & t_2    & \cdots    & 0        & 0       \\
            0      & t_2    & 0      & t_3       & \cdots   & 0        \\
            \vdots & \vdots & \vdots & \ddots    & \vdots   & \vdots    \\
            0      & 0      & \cdots &  t_{L-2} & 0        & t_{L-1}   \\
            t_{L} & 0      & \cdots & 0         & t_{L-1} & 0
        \end{pmatrix},
    \end{equation}
    and where we set $T_- = -T \1(T < 0)$. Here, $\1(T < 0)$ is the negative spectral projector of the self-adjoint matrix $T$, that is, the orthogonal projector on the direct sum of its negative eigenspaces. The first term in~\eqref{eq:def:Peierls_finite} is the distortion energy of the atoms: this energy depends quadratically on the distances $d_n$ between successive atoms, but these distances are themselves affine functions on the amplitudes $t_n$. The parameter $\mu > 0$ is the rigidity of the chain, and our units are such that the jump amplitude between two atoms is $1$ when their distortion energy is minimal. The second term in~\eqref{eq:def:Peierls_finite} models the electronic energy of the valence electrons under the Hamiltonian $T$. It results from the identity
    \begin{equation*} 
        \min_{0 \le \gamma = \gamma^* \le 1} 2 \Tr \left( T \gamma  \right) =  2 \Tr \left( T \1 (T < 0 )\right) = - 2 \Tr (T_-).
    \end{equation*}
    The minimization on the left-hand side was performed for all one-body density matrices $\gamma$ representing non-interacting electrons. The condition $0 \le \gamma \le 1$ is the Pauli principle, and the $2$ factor stands for the spin.
    \medskip
    
    We say that $\bt$ is a critical point of $\cE^{(L)}$ when one has $\cE^{(L)}(\bt+\bh) = \cE^{(L)}(\bt)+o(\Vert \bh\Vert)$.

    \subsection{The SSH energy difference}
    \label{ssec:Ft}
    
    Let us fix a configuration $\bt$, and consider the energy difference functional $\cF_\bt^{(L)}$, defined by
    \begin{equation}\label{Delta_Energy_t}
        \cF_\bt^{(L)}(\bh):= \cE^{(L)}(\bt+\bh) - \cE^{(L)}(\bt)= \frac{\mu}{2}\sum_{n\in\mathbb{Z}/L\mathbb{Z}}(h_n + 2t_n - 2)h_n - 2 \Tr( (T + H)_- - T_-),
    \end{equation}
    where $T = T(\bt)$ and $H = T(\bh)$ are the hermitian matrices constructed from $\bt$ and $\bh$ respectively. Clearly, $\bt $ is a critical point of $\cE^{(L)}$ iff $\bnull$ is a critical point of $\cF_\bt^{(L)}$. Substracting the quantity $\cE^{(L)} (\bt)$ allows us to define an analogous finite energy difference for infinite chains: for two bounded sequences $\bt : \Z \to \R^+$ and $\bh : \Z \to \R$, assuming that $h\in\ell^1(\Z,\R)$ and that $(T + H)_- - T_-$ is trace-class as an operator acting on $\ell^2(\Z,\C)$, we set
    \begin{equation} \label{eq:def:F}
        \boxed{ \cF_\bt(\bh):= \frac{\mu}{2}\sum_{n \in \Z} (h_n + 2t_n - 2)h_n - 2 \Tr( (T + H)_- - T_- ) } \ .
    \end{equation}
    Now, the operator $T := T(\bt) $ (and similarly for $T+H$) is acting on the infinite dimensional Hilbert space $\ell^2(\Z,\C)$. Its coefficients in the canonical basis are
    \[
    \forall n \in \Z, \quad T_{n, n+1} = T_{n+1, n}= t_n, \qquad T_{i,j} = 0 \quad \text{if $| i - j | \neq 1$}.
    \]
    In what follows, we denote by bold letters $\ba, \bt, \bh, \bu, ...$ sequences from $\Z$ to $\R$, and by capital letters $A, T, H, U, ...$ the corresponding operators acting on $\ell^2(\Z)$.
    
    The fact that the map $\cF_\bt$ is well defined when $\bt$ is a homoclinic or heteroclinic configuration is given in the next two lemmas (see Section~\ref{ssec:homoclinic_heteroclinic} for the proof).
    
    \begin{lemma} \label{lem:existence_contour}
        Let $\bt$ be a homoclinic or heteroclinic configuration such that $\bt \ge \tau$ for some $\tau > 0$. Then there is a positively oriented contour $\sC$ in the complex plane, a constant $C \ge 0$ and a constant $\eta > 0$ so that, for all $\bh$ with $ \| \bh \|_{\ell^\infty} \le  \eta$ and for all $z \in \sC$, the operator $(z - (T + H))$ is invertible with $\| (z - (T + H))^{-1} \|_{\rm op} \le C$ and
        \[
        -(T + H)_- = \frac{1}{2 \ri \pi} \oint_\sC \dfrac{z}{z - (T + H)} \rd z.
        \]
    \end{lemma}
    The contour $\sC$ is independent of $\bh$, but depends on $\bt$. This Lemma might be surprising, as the energy $0$ can be in the spectrum of $T+H$. Actually, we will prove the following:
    \begin{itemize}
        \item If $\bt$ is a homoclinic configuration with $\bt \ge \tau > 0$, then $0$ is never in the spectrum of $T +H$, for $\bh$ small enough. 
        \item If $\bt$ is a heteroclinic configuration with $\bt \ge \tau > 0$, then $0$ is always an isolated eigenvalue of $T+H$ of multiplicity $1$, for all $\bh$ small enough.
    \end{itemize}
    
    In both cases, one can choose a contour $\sC$ of the form (see Figure~\ref{fig:contours})
    \begin{equation} \label{eq:def:sC}
        \sC :=  (\Sigma + \ri) \to (\Sigma - \ri) \to (-g/2 - \ri) \to (-g/2 + \ri) \to (\Sigma + \ri),
    \end{equation}
    where $\Sigma$ is a negative enough number, and where $g = {\rm dist}(0, \sigma(T) \setminus \{ 0 \} )$ is the distance between $0$ and the (rest of the) spectrum.
    
    \medskip
    
    \begin{figure}[!ht]
        \centering
        \begin{subfigure}[t]{0.45\textwidth}
            \centering
            
            \begin{tikzpicture}[scale=0.6]
                
                \draw (-6,0)--(5.5,0);
                \filldraw[red]{} (-5.5,0) circle[radius = 0.07] node[left=0.3cm, above]{$\Sigma$};
                \filldraw[fill=blue!50] (-4.5,0.1) node[below=0.2cm]{$-2W$}rectangle(-2,-0.1) node[below=0.1cm]{$-2\delta$};
                \filldraw[fill=blue!50]  (-1.57,0) circle[radius=0.07];
                \filldraw[fill=blue!50](-1.03,0) circle[radius=0.07];
                \draw (0,-0.1) -- (0, 0.1) node[below=0.2cm]{$0$};
                \filldraw[red]{} (-0.5,0) circle[radius = 0.07] node[above=0.25cm, left=0.04cm]{$-\frac{g}{2}$};
                
                
                \coordinate (A) at (-5.5,0); 
                \coordinate (B) at (-5.5,-1); 
                \coordinate (C) at (-0.5,-1); 
                \coordinate (D) at (-0.5,1); 
                \coordinate (E) at (-5.5,1) ; 

                \draw [red](A) -- (B) -- (C) -- (D) -- (E) node[above , midway]{$\sC$} -- cycle ; 
                \draw[red, ->,>=stealth] (A) -- (-5.5,-0.5);
                \draw[red,->,>=stealth] (B) -- (-2,-1);
                \draw[red,->,>=stealth] (C) -- (-0.5,0.5);
                \draw[red,->,>=stealth] (D) -- (-3,1);

                
                \filldraw[fill=blue!50] (1.03,0)   circle[radius=0.07];  
                \filldraw[fill=blue!50] (1.57,0)  circle[radius=0.07];
                
                \filldraw[fill=blue!50] (2,-0.1) node[below=0.1cm]{$2\delta$} rectangle(4.5,0.1) node[below=0.2cm]{$2W$};
                \draw[->,>=stealth] (5.2,0)--(5.5,0)node[above]{$\sigma(T)$};
            \end{tikzpicture}
            
        \end{subfigure}
        \hfill
        \begin{subfigure}[t]{0.45\textwidth}
            \centering
            
            \begin{tikzpicture}[scale=0.6]
                
                \draw (-6,0)--(5.5,0);
                \filldraw[red]{} (-5.5,0) circle[radius = 0.07] node[left=0.3cm, above]{$\Sigma$};
                \filldraw[fill=blue!50] (-4.5,0.1) node[below=0.2cm]{$-2W$}rectangle(-2,-0.1) node[below=0.1cm]{$-2\delta$};
                \filldraw[fill=blue!50]  (-1.57,0) circle[radius=0.07];
                \filldraw[fill=blue!50](-1.03,0) circle[radius=0.07];
                \filldraw[fill=blue!50]  (0,0) circle[radius=0.07] node[below=0.2cm]{$0$};
                \filldraw[red]{} (-0.5,0) circle[radius = 0.07] node[above=0.25cm, left=0.04cm]{$-\frac{g}{2}$};
                
                
                \coordinate (A) at (-5.5,0); 
                \coordinate (B) at (-5.5,-1); 
                \coordinate (C) at (-0.5,-1); 
                \coordinate (D) at (-0.5,1); 
                \coordinate (E) at (-5.5,1) ; 

                \draw [red](A) -- (B) -- (C) -- (D) -- (E) node[above , midway]{$\sC$} -- cycle ; 
                \draw[red, ->,>=stealth] (A) -- (-5.5,-0.5);
                \draw[red,->,>=stealth] (B) -- (-2,-1);
                \draw[red,->,>=stealth] (C) -- (-0.5,0.5);
                \draw[red,->,>=stealth] (D) -- (-3,1);

                
                \filldraw[fill=blue!50] (1.03,0)   circle[radius=0.07];  
                \filldraw[fill=blue!50] (1.57,0)  circle[radius=0.07];
                
                \filldraw[fill=blue!50] (2,-0.1) node[below=0.1cm]{$2\delta$} rectangle(4.5,0.1) node[below=0.2cm]{$2W$};
                \draw[->,>=stealth] (5.2,0)--(5.5,0)node[above]{$\sigma(T)$};
            \end{tikzpicture}
        \end{subfigure}
        \caption{Contours used for the Cauchy integral, for a homoclinic configuration (Left), and a heteroclinic configuration (Right). The main difference is that $0$ is in the spectrum in the heteroclinic case. We prove below that $\sigma_{\rm ess}(T) = [-2W, -2\delta] \cup [2 \delta, 2W]$, and that the spectrum of $T$ is symmetric with respect to $0$.}
        \label{fig:contours}
    \end{figure}
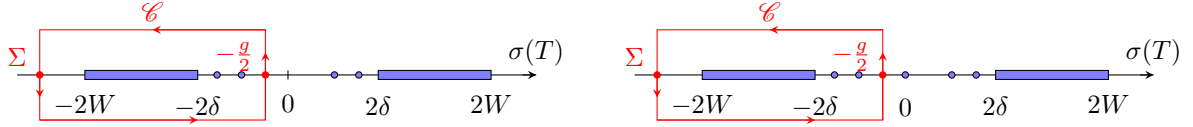

    \medskip

    In the heteroclinic situation, $0$ is a stable (or topologically protected) eigenvalue: it is unperturbed by the addition of $H$. Actually, any $T$ matrix coming from a heteroclinic configuration can be seen as a junction between two SSH chains with different indices~\cite{SuSchHee79, GonMonPer22, Kit09}.
    
    \medskip
    
    This Lemma allows to prove that $\cF_{\bt}$ is well-defined and smooth around $\{ \bnull \}$. We refer to Section~\ref{ssec:TaylorF} for the proof of the following result.
    \begin{lemma} \label{lem:F_is_smooth}
        Let $\bt$ be a homoclinic or heteroclinic configuration with $\bt \ge \tau$ for some $\tau > 0$, and let $\eta > 0$ and $\sC$ be a contour as in Lemma~\ref{lem:existence_contour}. The map $\bh \mapsto \cF_\bt(\bh)$ is well-defined and $C^\infty$ on $\{ \bh, \| \bh \|_{\ell_1} < \eta \}$. 
        In addition, there is $C \ge 0$ so that, for all $\bh$ with $\| \bh \|_{\ell^1} < \eta$, we have
        \[
        \left| \cF_{\bt}(\bh) - L_{\bt}(\bh) - \frac12 H_\bt(\bh, \bh) \right| \le C \| \bh \|_{\ell^2}^3,
        \]
        where $L_{\bt}$ (differential of $\cF_\bt$) is the continuous linear form on $\ell^1(\Z)$ defined by (we set $\Gamma_\bt := \1(T < 0)$ the negative spectral projector of $T$)
        \[
        L_{\bt}(\bh) := \mu \sum_{n \in \Z} (t_n - 1) h_n + 2 \Tr \left( \Gamma_\bt H \right), 
        \]
        and $H_\bt$ (hessian of $\cF_\bt$) is the bilinear form on $\ell^1(\Z)$ defined by
        \begin{equation} \label{eq:def:H}
            H_\bt(\bh, \bk) := \mu \sum_{n\in \Z} h_n k_n + 2 \Tr \left( \frac{1}{2 \ri \pi} \oint_{\sC}  H \dfrac{1}{z - T} K \dfrac{1}{z - T} \rd z \right) .
        \end{equation}
        In addition, the bilinear map $H_\bt$ can be extended continuously as a bilinear map on $\ell^2(\Z)$.
    \end{lemma}
    
    \subsection{Critical points for the infinite SSH model, and main result}
    \label{ssec:criticalPoints}
    
    We can now define the notion of critical points for the infinite SSH model.
    
    \begin{definition}
        Let $\bt$ be a homoclinic or heteroclinic configuration such that $\bt \ge \tau$ for some $\tau > 0$. We say that $\bt$ is a critical point if $L_\bt$ is the null map. Equivalently, using that
        \[
        \Tr(\Gamma_{\bt}H) = \sum_{n \in \Z} h_n \left[ ( \Gamma_\bt)_{n+1, n} + (\Gamma_{\bt})_{n,n+1}  \right] = 2\sum_{n\in\Z} h_n \left( \Gamma_\bt \right)_{n, n+1},
        \]
        the configuration $\bt$ is a critical point if 
        \begin{equation} \label{eq:def:critical}
            \forall n \in \Z, \quad  \boxed{ t_n = 1 - \frac{4}{\mu} \left( \Gamma_\bt \right)_{n, n+1} } \ .
        \end{equation}
    \end{definition}
    
    We implicitly used that $\Gamma$ is symmetric and real-valued. With this definition, the kink state constructed in~\cite{GarSer12} is a heteroclinic critical point (see Figure~\ref{fig:kink}). Now we can provide our main result, which states that all homoclinic or heteroclinic critical points of $\cF_\bt$ converge exponentially fast to $\bt^+$ at $+\infty$. 
    
    \begin{theorem}\label{th:main_result}
        Let $\bt$ be a homoclinic or heteroclinic critical point, and let $\bu$ be the sequence $u_n := t_n - t_n^+$. If $\bu$ is square integrable at $+\infty$ ({\it i.e.} $\sum_{n\geq 1} u_n^2<\infty$), then $\bu$ is exponentially localized at $+ \infty$: there are $C \ge 0$ and $\alpha > 0$ so that
        \[
        \left| u_n \right| \le C \re^{ - \alpha n}.
        \]
    \end{theorem}
    
    Of course, the same applies in the $-\infty$ direction, and we have exponential convergence to $\bt^+$ or $\bt^-$ at $-\infty$ depending whether the critical configuration is homoclinic or heteroclinic.
    \medskip

    \begin{figure}[!h]
        \includegraphics[scale = 0.7]{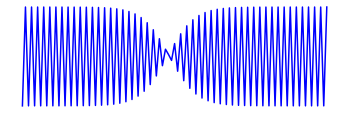}
        \caption{A localized kink appears in the chain with $L=101$ carbon atoms.}
        \label{fig:kink}
    \end{figure}

    Note that this exponential localization was already known for the exactly soluble continuum model of Takayama, Lin-Liu and Maki~\cite{TakLinMak80, Fei95, Fei04}.

    \subsection{Strategy of the proof}
    \label{ssec:strategy_proof}
    
    Let us briefly explain the strategy to prove Theorem~\ref{th:main_result}. We break the proof in several Lemmas, that we prove later in Section~\ref{sec:proofs}.
    
    \medskip

    Let $\bt$ be a homoclinic or heteroclinic critical point, and let $\{ \bu \}$ be the sequence $u_n := t_n - t_n^+$, so that $T = T^+ + U$. The configurations $\{\bt\}$ and $\{\bt^+\}$ are critical points, hence satisfy the Euler-Lagrange equations 
    \begin{equation*}
        t_n =  1 - \frac{4}{\mu}\G_{n,n+1}, \qquad 
        t_n^+ = 1 - \frac{4}{\mu} \G^+_{n,n+1},
    \end{equation*}
    with $\G:= \1(T^+ + U< 0)$, and $\G^+ := \1(T^+< 0)$. 
    According to Lemma~\ref{lem:existence_contour}, the expression of $\Gamma$ and $\Gamma^+$ can be written using Cauchy's residue formula with the {\em same} contour $\sC$, that is
    \[
    \G = \frac{1}{2\ri\pi}\oint_{\sC}\frac{\rd z}{ z - (T^+ + U) }, 
    \quad \text{and} \quad
    \G^+ = \frac{1}{2\ri\pi}\oint_{\sC}\frac{\rd z}{ z - T^+ },
    \]
    where the operators in the integrand are uniformly bounded in $z \in \sC$. Since $u_n = t_n - t_n^+$, we obtain (we use the resolvent formula in the last line)
    \begin{align*}\label{Euler_Equation_with_Cauchy}
        u_n & = \frac{4}{\mu} \left(\G^+ - \G\right)_{n,n+1}  = \frac{4}{\mu}  \left( \frac{1}{2\ri\pi}\oint_{\sC} \left[ \frac{1}{ z - T^+ } - \frac{1}{ z - (T^+ + U) } \right] \rd z \right)_{n, n+1} \\
        & = \frac{-4}{\mu}  \left( \frac{1}{2\ri\pi}\oint_{\sC} \left( \frac{1}{z-T^+}U\frac{1}{z - T^+} \right) \rd z \right)_{n, n+1} +  \frac{1}{\mu} (\bQ_{U}(\bu,\bu))_n,
    \end{align*}
    with the remainder term
    \[
    (\bQ_{U}(\bu_1, \bu_2))_n = - 4 \left(\frac{1}{2 \ri \pi}\oint_{\sC}\frac{1}{z-T^+}U_1\frac{1}{z - (T^++U)}U_2\frac{1}{z - T^+} \rd z\right)_{n,n+1}.
    \]
    Multiplying by $\mu$ and reordering the terms, this can be also written as 
    \begin{equation}\label{equation_L_Q}
        \forall n \in \Z, \quad (\sL\bu)_n = (\bQ_U(\bu,\bu))_n, 
    \end{equation}
    with the linear map
    \begin{equation} \label{eq:def:sL}
        (\sL\bu)_n = \mu u_n +  4 \left( \frac{1}{2\ri\pi}\oint_{\sC} \left( \frac{1}{z-T^+}U\frac{1}{z - T^+} \right) \rd z \right)_{n, n+1}.
    \end{equation}
    
    Formally, if $\bv$ is another real sequence, with corresponding operator $V$, we have
    \begin{equation} \label{eq:vLu}
        \langle \bv, (\sL \bu ) \rangle = \sum_{n \in \Z} v_n (\sL \bu )_n = \mu \sum_{n \in \Z} v_n u_n + 2 \Tr \left( \frac{1}{2\ri\pi}\oint_{\sC}  \frac{1}{z-T^+}U\frac{1}{z - T^+} V \right) \rd z
    \end{equation}
    and we recognize the expression of the Hessian $H_{\bt^+}(\bv, \bu)$ in~\eqref{eq:def:H}. Unfortunately, the previous computations is formal, since $\bu$ is not necessary in $\ell^2(\Z)$. We only know that it is square integrable at $+\infty$. Actually, for a heteroclinic configuration, we have $\bu \notin \ell^2(\Z)$, since $\bu$ does not decay to $0$ at $-\infty$. In order to bypass this difficulty, we regularize $\bu$ using appropriate cut-off functions.

    \medskip

    For $\alpha > 0$ and $s \in \Z$ that we choose later ($\alpha$ will be small, and $s$ will be large), we introduce the function $\theta_{\alpha, s} : \Z \to \R^+$ defined by 
    \begin{equation}\label{theta_n}
        \theta_{\alpha, s}(n) = \min \{ \re^{\alpha n}, \re^{ \alpha s} \} = 
        \displaystyle
        \begin{cases}
            \re^{\alpha n}, \; \text{ if } n<s \cr
            \re^{\alpha s}, \; \text{ if } n\ge s
        \end{cases},
    \end{equation}
    and denote by $\Theta_{\alpha, s}$ the multiplication operator by $\theta_{\alpha, s}$, defined by $(\Theta_{\alpha, s})_{n,m} = \theta_{\alpha, s}(n)\delta_{n,m}$.
    
    \medskip
    
    In what follows, we will consider the sequence $\widetilde{\bu}_{\alpha, s}$, defined by
    \[
    (\widetilde{u}_{\alpha, s})_n := \theta_{\alpha, s}(n) \theta_{\alpha, s}(n+1)  u_n,
    \quad \text{with corresponding operator}  \quad
    \widetilde{U}_{\alpha, s} = \Theta_{\alpha, s} U \Theta_{\alpha, s}.
    \]
    Since $\bu$ is bounded and square integrable at $+\infty$ the vector $\widetilde{\bu}_{\alpha, s}$ is in $\ell^2(\Z)$ for all $\alpha > 0$ and all $s \in \Z$. We also introduce the operator $\widetilde{T}_{\alpha, s}^+$ acting on $\ell^2(\Z)$, and defined in the canonical basis by
    \[
    \forall n \in \Z, \qquad 
    \left( \widetilde{T}_{\alpha, s}^+ \right)_{n, n+1} := \dfrac{\theta_{\alpha, s}(n)}{\theta_{\alpha, s}(n+1)} t_n^+ , \quad 
    \left( \widetilde{T}_{\alpha, s}^+ \right)_{n+1, n} := \dfrac{\theta_{\alpha, s}(n+1)}{\theta_{\alpha, s}(n)} t_n^+ ,
    \]
    and $\left( \widetilde{T}_{\alpha, s}^+ \right)_{i,j} = 0$ if $| i -j | \neq 1$. 
    Note that $\widetilde{T}_{\alpha, s}^+$ is not symmetric. Using that
    \[
    \dfrac{\theta_{\alpha, s}(n)}{\theta_{\alpha, s}(n+1)} = 
    \begin{cases}
        \re^{- \alpha} \quad \text{if} \quad n < s \\
        1 \quad \text{if} \quad n \ge s,
    \end{cases}
    \]
    we see that $\widetilde{T}_{\alpha, s}^+$ has the matrix form
    \begin{equation}
        \widetilde{T}_{\alpha,s}^+ = 
        \begin{pmatrix}
            \ddots & \ddots    & \ddots        & \ddots & \ddots &  \ddots& \ddots\\
            \ddots & 0    & t_{s-2}^+ \re^{-\alpha}     & 0       & 0     & 0        & \ddots\\
            \ddots & t_{s-2}^+ \re^{\alpha}      & 0    & t_{s-1}^+ \re^{-\alpha}        & 0         & 0        & \dots \\
            \ddots    & 0       & t_{s-1}^+ \re^{\alpha}       & 0    & t_s^+     & 0    & \ddots\\
            \ddots  & 0   & 0    & t_s^+   & 0         & t_{s+1}^+        & \ddots\\
            \ddots & 0 & 0 & 0 & t_{s+1}^+     & 0        & \ddots   \\
            \ddots & \ddots & \ddots & \ddots & \ddots     & \ddots   & \ddots
        \end{pmatrix}.
    \end{equation}
    
    This operator is constructed to satisfy the following commutation relations (see Section~\ref{ssec:proof:Talphas} for the proof).
    \begin{lemma} \label{lem:Talphas}
        The operator $\widetilde{T}_{\alpha,s}^+$ satisfies
        \begin{equation} \label{eq:key_equality_T_Theta}
            \Theta_{\alpha, s} T^+ = \widetilde{T}_{\alpha,s}^+ \Theta_{\alpha, s} \quad \text{and} \quad
            T^+ \Theta_{\alpha, s}  = \Theta_{\alpha, s} \left( \widetilde{T}_{\alpha,s}^+  \right)^*.
        \end{equation}
        There is $\alpha_* > 0$ and $C \ge 0$ so that, for all $0 \le \alpha < \alpha_*$, all $s \in \Z$, and all $z \in \sC$, the operators $z - \widetilde{T}_{\alpha,s}^+$ and $z - (\widetilde{T}_{\alpha,s}^+)^*$ are invertible, with $\|(z - \widetilde{T}_{\alpha,s}^+)^{-1} \|_{\rm op} \le C$ and $\|(z - (\widetilde{T}_{\alpha,s}^+)^*)^{-1} \|_{\rm op} \le C$. In addition, we have
        \[
        \Theta_{\alpha, s}\dfrac{1}{z - T^+} = \dfrac{1}{z - \widetilde{T}_{\alpha,s}^+}\Theta_{\alpha, s}  , 
        \quad \text{and} \quad
        \dfrac{1}{z - T^+} \Theta_{\alpha, s} = \Theta_{\alpha, s} \dfrac{1}{z - (\widetilde{T}_{\alpha,s}^+)^*} .
        \]
    \end{lemma}

    We multiply~\eqref{equation_L_Q} on the left by $\theta_s(n)$, and on the right by $\theta_s(n+1)$. Using that, for any operator $A$ on $\ell^2(\Z)$, we have $\theta_{\alpha, s}(n) A_{n,n+1} \theta_{\alpha, s}(n+1) = (\Theta_{\alpha, s} A \Theta_{\alpha, s})_{n,n+1}$, and the fact that
    \begin{align*}
        \Theta_{\alpha, s}  \frac{1}{z-T^+}U\frac{1}{z - T^+} \Theta_{\alpha, s} 
        = \frac{1}{z-\widetilde{T}^+_{\alpha, s}} \underbrace{\Theta_{\alpha, s} U \Theta_{\alpha, s}}_{ = \widetilde{U}_{\alpha, s}} \frac{1}{z - (\widetilde{T}^+_{\alpha, s})^*},
    \end{align*}
    we obtain an equation of the form
    \begin{equation} \label{eq:widetilde_LU=Q}
        \left( \widetilde{\sL}_{\alpha, s} \widetilde{\bu}_{\alpha, s} \right)_n = \left( \widetilde{\bQ}_{\alpha, s, U}(\bu,\bu) \right)_n,
    \end{equation}
    where $\widetilde{\sL}_{\alpha, s}$ is the operator defined on $\ell^2(\Z)$ by
    \begin{equation} \label{eq:def:L_alpha_s}
        \forall \widetilde{\bv} \in \ell^2(\Z), \quad  \left( \widetilde{\sL}_{\alpha, s} \widetilde{\bv} \right)_n := \mu  (\widetilde{v}_{\alpha, s})_n + 4 \left( \frac{1}{2\ri\pi}\oint_{\sC} \left( \frac{1}{z-\widetilde{T}^+_{\alpha, s}} \widetilde{V} \frac{1}{z - (\widetilde{T}^+_{\alpha, s})^*} \right) \rd z \right)_{n, n+1},
    \end{equation}
    with the right-hand side given by
    \begin{equation} \label{eq:def:Q}
        (\widetilde{\bQ}_{\alpha, s, U}(\bu_1,\bu_2))_n = - 4 \left(\frac{1}{2 \ri \pi}\oint_{\sC}\frac{1}{z-\widetilde{T}_{\alpha, s}^+} (\Theta_{\alpha, s} U_1) \frac{1}{z - (T^+ +U)} (U_2 \Theta_{\alpha, s}) \frac{1}{z - (\widetilde{T}_{\alpha, s}^+)^*} \rd z\right)_{n,n+1}.
    \end{equation}
    
    The exponential decay is a consequence of the following Lemmas.
    
    \begin{lemma} \label{lem:sL_invertible}
        The operator $\sL$ defined in~\eqref{eq:def:sL}, seen as an operator from $\ell^2(\Z)$ to itself, is bounded symmetric with bounded inverse. \\
        There is $\alpha_* > 0$ and $C \ge 0$ so that, for all $0 \le \alpha < \alpha_*$ and all $s \in \Z$, the operator $\widetilde{\sL}_{\alpha, s}$ defined in~\eqref{eq:def:L_alpha_s}, seen as an operator from $\ell^2(\Z)$ to itself, is bounded with bounded inverse.
    \end{lemma}
    
    Note that the operator $\widetilde{\sL}_{\alpha, s}$ is not symmetric for $\alpha > 0$. We refer to Section~\ref{ssec:proof:sL_invertible} for the proof. A key property that we use in the proof is the fact that the Hessian $H_{\bt^+}$ is coercive (see Proposition~\ref{prop-coercive} below). Due to the equality $\bra \bv, \sL \bv \ket_{\ell^2} = H_{\bt^+}(\bv, \bv)$ for $\bv \in \ell^2(\Z)$ (see~\eqref{eq:vLu}), this implies that $\sL$ is invertible.
    
    In order to control the right-hand side of~\eqref{eq:widetilde_LU=Q}, we use the following result (see Section~\ref{ssec:proof:Q} for the proof).
    
    \begin{lemma} \label{lem:Q}
        There is $\alpha_* > 0$ and $C \ge 0$ so that, for all $0 \le \alpha < \alpha_*$ and all $s \in \Z$, we have 
        \[
        \forall \bu_1, \bu_2 \in \ell^\infty(\Z) \cap \ell^2(\Z^+), \qquad
        \left\| \widetilde{Q}_{\alpha, s, U}(\bu_1, \bu_2) \right\|_{\ell^2(\Z)} \le C \| \theta_{\alpha, s} \bu_1 \|_{\ell^4} \| \theta_{\alpha, s} \bu_2 \|_{\ell^4}.
        \]
    \end{lemma}
    
    We can now prove the exponential decay of $\bu$. From~\eqref{eq:widetilde_LU=Q} and the two Lemmas, we get that there is $C \ge 0$ and $\alpha_* > 0$ so that, for all $0 < \alpha \le \alpha_*$ and all $s \in \Z$, we have
    \begin{equation} \label{eq:l2_le_l4}
        \| \widetilde{\bu}_{\alpha, s} \|_{\ell^2}^2 \le C \| \theta_{\alpha, s} \bu \|_{\ell^4}^4.
    \end{equation}
    Concerning the left-hand side, we note that $\theta_{\alpha, s}$ is increasing so that $\theta_{\alpha, s}(n) \theta_{\alpha, s}(n+1) \ge \theta_{\alpha, s}^2(n)$. Hence
    \[
    \| \theta_{\alpha, s}^2 \bu \|_{\ell^2}^2 \le \| \widetilde{\bu}_{\alpha, s} \|_{\ell^2}^2.
    \]
    Let us now bound the right-hand side. We fix $\varepsilon  := \frac{1}{\sqrt{2C}}$, where $C$ is the constant appearing in~\eqref{eq:l2_le_l4}. Since $\bu$ goes to $0$ at $+ \infty$, there is $M$ large enough so that $ | u_n |  < \varepsilon$ for all $n \ge M$. This gives
    \begin{align*}
        \| \theta_{\alpha, s} \bu \|_{\ell^4}^4 
        & =  \sum_{n \le M} \theta_{\alpha, s}^4(n) | u_n |^4 + \sum_{n > M }  \theta_{\alpha, s}^4(n) | u_n |^4 
        \le \| \bu \|_{\ell^\infty}^4 \sum_{n \le M} \theta_{\alpha, s}^4(n)
        + \varepsilon^2 \sum_{n > M} \theta_{\alpha, s}^4(n) | u_n |^2 \\
        & = \| \bu \|_{\ell^\infty}^4 \sum_{n \le M} \re^{4\alpha n} + \varepsilon^2 \| \theta_{\alpha, s}^2 \bu \|_{\ell^2}^2
        = \| \bu \|_{\ell^\infty}^4 \dfrac{\re^{4\alpha M}}{1 - \re^{-4\alpha}} + \varepsilon^2 \| \theta_{\alpha, s}^2 \bu \|_{\ell^2}^2.
    \end{align*}
    Plugging these inequalities in~\eqref{eq:l2_le_l4} gives
    \[
    (1 - C \varepsilon^2) \| \theta_{\alpha, s}^2 \bu \|_{\ell^2}^2 \le \| \bu \|_{\ell^\infty}^4 \dfrac{\re^{4\alpha M}}{1 - \re^{-4\alpha}}.
    \]
    With our choice of $\varepsilon$, the quantity $1 - C \varepsilon^2 = \frac12$ is positive. The right-hand side is a bound independent of $s \in \Z$. We can take the limit $s \to \infty$, and conclude that
    \[
    \left( \re^{ 2 \alpha n} u_n \right)_{n \in \Z} \in  \ell^2(\Z), \quad \text{with} \quad \left\| \left( \re^{ 2 \alpha n} u_n\right)_{n \in \Z} \right\|_{\ell^2(\Z)} \le 2 \| \bu \|_{\ell^\infty}^4 \dfrac{\re^{4\alpha M}}{1 - \re^{-4\alpha}}.
    \]
    This proves as wanted that the sequence $\bu$ is exponentially decaying at $+ \infty$.

    \section{Smoothness and Taylor expansion of the map $\cF_\bt$}
    
    In this section, we prove Lemma~\ref{lem:existence_contour}, which states that $\bh \mapsto (T+H)_-$ is smooth locally around $\bnull$, whenever $T$ is a homoclinic or heteroclinic configuration. 
    
    \medskip
    
    We first record a useful Lemma that we will use many times throughout the article. In what follows, we denote by $\cB := \cB(\ell^2(\Z))$ the set of bounded operators acting in $\ell^2(\Z)$, and by $\fS_p := \fS_p(\ell^2(\Z))$ the $p$--Schatten class: $A \in \fS_p$ iff A is a compact operator with $ \| A \|_{\fS_p} := \Tr( | A |^p )^{1/p} < + \infty$. The set $\fS_\infty$ is simply the set of compact operators, with $\| A \|_{\fS_\infty} = \| A \|_{\rm op}$.

    \begin{lemma} \label{lem:a_and_A}
        Let $\ba$ be a sequence from $\Z$ to $\R$, and let $A$ be the corresponding operator.
        \begin{itemize}
            \item If $\ba \in \ell^\infty$, then $A$ is a bounded operator ($A \in \cB$), and $\| A \|_{\rm op} \le 2 \| \ba \|_{\ell^\infty}$ ;
            \item If $\ba$ goes to $0$ at $\pm \infty$, then $A$ is compact ($A \in \fS_\infty$) ;
            \item If $\ba \in \ell^p(\Z)$ for some $1 \le p < \infty$, then $A$ is in the Schatten class $\fS_p$, and
            \[
            \| A \|_{\fS_p} \le 2 \| \ba \|_{\ell^p}.
            \]
        \end{itemize}
    \end{lemma}

    \begin{proof}
        For the first part, we note that, for all $\psi \in \ell^2(\Z)$, we have
        \begin{equation*} 
            \left| \langle \psi, A \psi \rangle_{\ell^2} \right| = \left| \sum_{n \in \Z} a_n (\overline{\psi_n} \psi_{n+1} + \overline{\psi_{n+1}} \psi_n ) \right|
            \le \| \ba \|_{\ell^\infty} \sum_{n \in \Z} \left( |\psi_n|^2 + |\psi_{n+1} |^2 \right) = 2 \| \ba \|_{\ell^\infty} \| \psi \|_{\ell^2}^2,
        \end{equation*}
        where we used that $\overline{a}b + a \overline{b} \le | a |^2 + | b |^2$ in the middle inequality. \\
        For the second part, we note that the operator $A$ is the limit, for the operator norm, of the finite-rank operators $A^N$ associated with the truncated configurations $a^N:=(\mathbf{1}_{-N\leq n \leq N}\, a_n)_{n\in \Z}$. Hence $A$ is compact. \\
        For the last part, we first prove the result for $p = 1$. We have, by duality,
        \begin{align*}
            \| A \|_{\fS_1} & = \sup_{K \in \cB \atop \| K \|_{\rm op} = 1} | \Tr(A K) | 
            =  \sup_{K \in \cB \atop \| K \|_{\rm op} = 1} \left| \sum_{n \in \Z} a_n (K_{n+1, n} + K_{n, n+1}) \right|\\ 
            & \le \| \ba \|_{\ell^1} \sup_{K \in \cB \atop \| K \|_{\rm op} = 1} \sup_{n \in \Z}\, (|K_{n+1, n}|  + | K_{n, n+1} |) \le 2 \| \ba \|_{\ell^1}.
        \end{align*}
        We used in the last line that $| K_{n,n+1}| = | \langle e_n, K e_{n+1} \rangle | \le \| K \|_{\rm op}$. Finally, to conclude the proof, we proceed by interpolation using Riesz-Thorin interpolation theorem for Schatten spaces (see~\cite[Remark 1 p.23]{Sim05} and~\cite[p.115]{ReeSim-77} for the version with $\cB$ instead of $\fS_\infty$).
    \end{proof}

    \subsection{The spectrum of homoclinic and heteroclinic configurations}
    \label{ssec:homoclinic_heteroclinic}
    
    In order to prove that $\cF_\bt$ is smooth, we first study the spectrum of the operator $T$ when $\bt$ is such a configuration. We treat the two cases separately.
    
    \subsubsection{Homoclinic configurations} Let $\bt$ be a homoclinic configuration with $\bt \ge \tau$ for some $\tau > 0$. Then, we can write $\bt = \bt^+ + \bu$, where we recall that $\bt^+$ is the dimerized configuration $t_n^+ = W + (-1)^n \delta$ with $\delta > 0$, and where the sequence $\bu$ goes to $0$ at $\pm \infty$.
    
    \medskip
    
    We have $T = T^+ + U$. The operator $T^+$ has purely essential spectrum, of the form (see for instance~\cite{GonKouSer23} and references therein)
    \[
    \sigma \left( T^+ \right) = \sigma_{\rm ess} \left( T^+ \right) = [-2W, -2 \delta] \cup [2 \delta, 2W].
    \]
    In particular, $T^+$ has a spectral gap of size $4 \delta$ around $0$. On the other hand, since $\bu$ goes to $0$ at $\pm \infty$, $U$ is compact, see Lemma~\ref{lem:a_and_A}. We thus deduce from Weyl's theorem that 
    \begin{equation} \label{eq:equality_spectra}
        \sigma_{\rm ess}(T) = \sigma_{\rm ess}(T^+) = [-2W, -2 \delta] \cup [2 \delta, 2W].
    \end{equation}
    In particular, $0 \notin \sigma_{\rm ess}(T)$. In addition, we claim that $0$ is not an eigenvalue of $T$. More specifically, we have the following.
    
    \begin{lemma} \label{lem:zero_mode}
        Let $\bt$ be {\bf any} configuration with $\bt \ge \tau$ for some $\tau > 0$ (in particular, all coefficients $t_n$ are non null). Assume there is $N_0 \in \N$ and $0 < \kappa < 1$ so that
        \[
        \text{\bf (Homoclinic case)} \quad \forall n \ge N_0, \quad \left| \frac{t_{2n+1}}{t_{2n}} \right| \le \kappa\quad \text{and} \quad
        \left| \frac{t_{-2n-1}}{t_{-2n}} \right| \le \kappa.
        \]
        Then $0$ is not an eigenvalue of $T$. Conversely, if
        \[
        \text{\bf (Heteroclinic case)} \quad \forall n \ge N_0, \quad \left| \frac{t_{2n+1}}{t_{2n}} \right| \le \kappa
        \quad \text{and} \quad
        \left| \frac{t_{-2n}}{t_{-2n+1}} \right| \le \kappa,
        \]
        then $0$ is an eigenvalue of $T$ of multiplicity $1$.
    \end{lemma}
    
    For a homoclinic (resp. heteroclinic) configuration, the first (resp. second) condition is satisfied with $\kappa = \frac{W - \delta/2}{W + \delta/2} < 1$.
    
    \begin{proof}
        The eigenvalue equation $T \psi = 0$ reads 
        \[
        \forall n\in \Z, \quad t_n \psi_{n} + t_{n+1} \psi_{n+2} = 0.
        \]
        For $n>0$ we obtain directly
        \[
        \psi_{2n}= (-1)^n \prod_{m=1}^{n} \left(\frac{t_{2m-2}}{t_{2m-1}}\right) \psi_0, \quad
        \psi_{2n + 1} = (-1)^n \prod_{m=1}^{n}\left(\frac{t_{2m - 1} }{t_{2m}}\right) \psi_1,
        \]
        and similar formulas hold for $n<0$. The vector space $ \{\psi, \ T \psi = 0\}$ is therefore two-dimensional, since a solution of the equation $T \psi = 0$ can be recovered from its values $\psi_0$ and $\psi_1$. The question is to find $\Ker(T) = \{T \psi = 0\}\cap\ell^2(\Z)$. 
        
        \medskip
        
        Let us first consider the homoclinic case, and let $\psi \in \{ T \psi = 0 \}$. Since $|t_{2n}/t_{2n+1}| \ge \kappa^{-1} > 1$ for $n \ge N_0$, we have $| \psi_{2N_0+2k} | \ge | \psi_{2 N_0} | \kappa^{-k}$ as $k \to \infty$, so $\psi$ cannot be square integrable at $+ \infty$, unless $\psi_{2 N_0} = 0$, which is equivalent to $\psi_0 = 0$. Similarly, we have $| \psi_{-2N_0 - 2k + 1} | \ge | \psi_{-2 N_0 + 1} | \kappa^{-k}$ as $k \to \infty$, so $\psi$ cannot be square integrable at $-\infty$, unless $\psi_{-2 N_0 + 1} = 0$, which gives $\psi_1 = 0$ as well. So $\Ker(T) = \{0\}$.
        
        \medskip
        
        In the heteroclinic case, the same reasoning shows that we must have $\psi_0 = 0$. However, given $\psi_1 \in \R$, the function $\psi$ with $\psi_{2n+1} = (-1)^n \prod_{m=1}^{n}\left(\frac{t_{2m - 1} }{t_{2m}}\right) \psi_1$ for $n$ positive, $\psi_{2n+1} = (-1)^n \prod_{m=1}^{\vert n\vert}\left(\frac{t_{-2m +2} }{t_{-2m+1}}\right) \psi_1$ for $n$ negative and $\psi_{2n} = 0$ for all $n$, is a square integrable non null eigenvector. In this case, $\dim \Ker(T) = 1$.
    \end{proof}
    
    \begin{remark} \label{rem:zero_mode}
        In the heteroclinic case, the corresponding normalized eigenvector $\psi$ is sometimes called an {\em edge state}, or {\em interface state} or {\em zero mode}. As shown in the proof, it is exponentially decaying at $\pm \infty$: there is $C \ge 0$ and $\beta := - \log(\kappa) > 0$ so that $| \psi_n | \le C \re^{- \beta | n |}$. It is always exponentially decaying, even though the sequence $\bt$ may converge to $\bt^\pm$ very slowly at $\pm \infty$. Actually, we do not require $\bt$ to be a critical point here. 
        
        Note that it is only supported on the odd integers: $\psi_{2n} = 0$ for all $n \in \Z$. In particular, the corresponding projector $Z := | \psi \ket \bra \psi |$ satisfies
        \[
        \forall n \in \Z, \qquad Z_{n, n+1} = Z_{n+1, n}= 0.
        \]
    \end{remark}
    
    Let us return to the homoclinic case. We proved that $0 \notin \sigma(T)$. Let $g := {\rm dist}(0, \sigma(T))$ be the distance between $0$ and the spectrum of $T$, and set $\eta := g/8$. Let $\bh$ be any perturbation with $\| \bh \|_\infty \le \eta$. Then $\| H \|_{\rm op} \le 2 \eta$ by Lemma~\ref{lem:a_and_A}. In particular, the spectrum of $T + H$ is $2\eta$--close to the one of $T$, hence $\sigma(T + H) \cap (-3g/4, 3g/4) = \emptyset$. 
    
    \medskip
    
    Let us consider the positively oriented contour $\sC$ in~\eqref{eq:def:sC}. We deduce first that $(z - (T+H))$ is invertible for all $z \in \sC$, and with $\| (z - (T+H))^{-1} \|_{\rm op} \le C$ for a constant $C$ independent of $z \in \sC$. Also, from the Cauchy residual formula, we have
    \[
    \1(T + H < 0)  = \frac{1}{2 \ri \pi} \oint_{\sC} \dfrac{\rd z}{z - (T+H)}, 
    \]
    and
    \[
    (T + H)_- = - (T+H) \1(T + H < 0)  = \frac{-1}{2 \ri \pi} \oint_{\sC} \dfrac{z}{z - (T+H)} \rd z. 
    \]

    \subsubsection{Heteroclinic configurations} \label{Heteroclinic_configurations}
    
    Now, let $\bt$ be a heteroclinic configuration with $\bt \ge \tau$ for some $\tau > 0$. First, we claim that $0 \notin \sigma_{\rm ess}(T)$.

    \begin{lemma}
        Let $\bt$ be a heteroclinic configuration. Then 
        \[
        \sigma_{\rm ess}(T) = [-2W, -2 \delta] \cup [2 \delta, 2 W].
        \]
        In particular, $0 \notin \sigma_{\rm ess}(T)$.
    \end{lemma}
    
    \begin{proof}
        Introduce the sequence $\widetilde{\bt}$ with $\widetilde{t_n} = t_n$ if $n \neq 0$, and $\widetilde{t_0} = 0$. We denote by $\widetilde{T}$ the corresponding operator, and set $K := T - \widetilde{T}$. Then $T = \widetilde{T} + K$ with
        {  \[  \arraycolsep=3pt \def\arraystretch{0.7}
           \widetilde{T} = 
            \left(
            \begin{array}{ccc|ccc}
                \ddots   & \ddots    &  0     &          &    &  \\
                \ddots    & 0      & t_{-1}    & 0    &         &        \\
                0      & t_{-1}    & 0      & 0       &  0  &         \\
                \hline 
                & 0 & 0 & 0    & t_1   & 0    \\
                &       & 0  &  t_{1} & 0        & \ddots   \\
                &       &  & 0         & \ddots & \ddots
            \end{array}
            \right)
            ,\quad K=
            \left(
            \begin{array}{ccc|ccc}
                \ddots   & \ddots    &  0     &          &    &  \\
                \ddots    & 0      & 0    & 0    &         &        \\
                0      & 0   & 0      & t_0       &  0  &         \\
                \hline 
                & 0 & t_0 & 0    & 0   & 0    \\
                &       & 0  &  0 & 0        & \ddots   \\
                &       &  & 0         & \ddots & \ddots
            \end{array}
            \right).
            \]}
        \medskip
        
        The operator $K$ is of rank $2$, hence is compact, so $\sigma_{\rm ess}(T) = \sigma_{\rm ess}(\widetilde{T})$ by Weyl's theorem. In addition, the operator $\widetilde{T}$ is of the form $\widetilde{T} = \widetilde{T}_L \oplus \widetilde{T}_R$ acting on $\ell^2(\Z) \approx \ell^2(\Z^-) \oplus \ell^2(\Z^+)$, hence
        \[
        \sigma_{\rm ess}(\widetilde{T}) = \sigma_{\rm ess}(\widetilde{T}_L) \cup \sigma_{\rm ess}(\widetilde{T}_R).
        \]
        Let us first focus on the right operator $\widetilde{T}_R$. The hopping amplitudes $\widetilde{t_n}$ for $n \ge 1$ are of the form $\widetilde{t_n} = t_n^+ + u_n$ with $\lim_{n \to \infty} u_n = 0$. So, with obvious notation, $\widetilde{T}_R = \widetilde{T}^+_R + U_R$. The sequence $(u_n)$ goes to zero, so $U_R$ is a compact operator (the proof is similar than for Lemma~\eqref{lem:a_and_A}), and $\sigma_{\rm ess}(\widetilde{T}_R) = \sigma_{\rm ess}(\widetilde{T}^+_R)$. Finally, reasoning as before and introducing the cut compact operator $K^+ := T^+ - \widetilde{T}^+$, we have
        \[
        \sigma(T^+) = \sigma_{\rm ess}(T^+) = \sigma_{\rm ess}(\widetilde{T}^+_L) \cup \sigma_{\rm ess}(\widetilde{T}^+_R).
        \]
        In addition, since $t_{-n}^+ = t_n^+$ for the dimerized configuration $\bt^+$, $\widetilde{T}^+_L$ is unitarily equivalent to $\widetilde{T}^+_R$, and in particular 
        \[
        \sigma_{\rm ess}(\widetilde{T}^+_R) = \sigma_{\rm ess}(\widetilde{T}^+_L) = [-2W, -2 \delta] \cup [2 \delta, 2 W].
        \]
        Altogether, we proved that $\sigma_{\rm ess}(\widetilde{T}_R) = [-2W, -2 \delta] \cup [2 \delta, 2 W]$.
        The proof for the left part is similar, upon replacing $T^+$ by $T^-$.
    \end{proof}
    
    \medskip
    
    In addition, using Lemma~\ref{lem:zero_mode}, we know that $0$ is an eigenvalue of $T$ of multiplicity $1$. So $0$ is an isolated eigenvalue, and we set 
    \[
    g := {\rm dist} \left( 0, \sigma(T) \setminus \{ 0 \} \right) \, >\, 0
    \quad \text{and} \quad
    \eta := \min \left\{  \frac{g}{8}, \frac{\tau}{2}, \frac{\delta}{2} \right\} > 0.
    \]
    By standard perturbation theory, for all $\bh$ with $\| \bh \|_{\ell^\infty} \le \eta$ (hence $\| H \|_{\rm op} \le 2 \eta$ by Lemma~\ref{lem:a_and_A}), the spectrum of $T + H$ is composed of an isolated eigenvalue $\lambda_0(T+H)$ of multiplicity $1$, with $| \lambda_0(T + H) | \le 2 \eta \le g/4$ corresponding to the perturbation of the $0$ eigenvalue of $T$, and the rest of the spectrum, at distance at least $g - 2 \eta > 3g/4$ from $0$.
    
    \medskip
    
    Since $\| \bh \|_{\ell^\infty} < \tau/2$ and $\| \bh \|_{\ell^\infty} < \delta/2$, the vector $\bt + \bh$ satisfies $\bt + \bh \ge \tau/2 > 0$ and $(t_n + h_n) \in (t_n - \delta/2, t_n + \delta/2)$. In particular, it satisfies the assumption of Lemma~\ref{lem:zero_mode} (heteroclinic case) with $\kappa = \frac{W - \delta/2 }{W + \delta/2} < 1$. So $\lambda_0(T + H) = 0$: the eigenvalue $0$ is unperturbed by the addition of $H$.
    
    \medskip
    
    We consider the positively oriented contour $\sC$ defined in~\eqref{eq:def:sC}. We deduce from the previous discussion that, for all $\bh$ with $\| \bh \|_{\ell^\infty} \le \eta$, we have
    \[
    (T + H)_- = - (T+H) \1(T + H < 0)  = \frac{-1}{2 \ri \pi} \oint_{\sC} \dfrac{z}{z - (T+H)} \rd z,
    \]
    where all operators appearing are uniformly bounded by some constant $C \ge 0$ independent of $z \in \sC$. We also remark that we have
    \[
    \1(T + H < 0) = \frac{1}{2 \ri \pi} \oint_{\sC} \dfrac{\rd z}{z - (T+H)},
    \quad \text{and} \quad
    \1(T + H \le 0) = \1 (T + H < 0) + Z,
    \]
    where $Z = | \psi \rangle \langle \psi |$ is the rank--$1$ projector onto the normalized zero-mode $\psi \in {\rm Ker}(T + H)$, see Remark~\ref{rem:zero_mode}.

    \subsection{Taylor expansion of the map $\cF_\bt$}
    \label{ssec:TaylorF}
    
    In this section, we study the energy $\cF_\bt$ in~\eqref{eq:def:F}, and prove Lemma~\ref{lem:F_is_smooth}. Recall that $\cF_\bt$ is defined by
    \begin{equation*}
        \cF_\bt(\bh):= \frac{\mu}{2}\sum_{n \in \Z} (h_n + 2t_n - 2)h_n - 2 \Tr( (T + H)_- - T_- ).
    \end{equation*}
    In what follows, $\bt$ is a homoclinic or heteroclinic configuration with $\bt \ge \tau$ for some $\tau > 0$. We introduce the constant $\eta > 0$ and the contour $\sC$ as in the previous section. 
    
    \medskip
    
    First, we claim that for all $\bh$ with $\| \bh \|_{\ell^1} \le \eta$ (we now use the $\ell^1$ norm), the map $\cF_\bt(\bh)$ is well-defined and finite. Since $\| \bh \|_{\ell^\infty } \le \| \bh \|_{\ell^1} \le \eta$, $\bh$ satisfies the conditions of the previous section. For the first part of the energy, we write that
    \[
    \sum_{n \in \Z} h_n^2 = \| \bh \|_{\ell^2}^2 \le \| \bh \|_{\ell^1}^2, 
    \quad \text{and} \quad
    \left| \sum_{n \in \Z} (2t_n - 2) h_n \right| \le 
    \left( 2 \| \bt \|_{\ell^\infty} + 2 \right) \| \bh \|_{\ell^1},
    \]
    so the first part is continuous from $\ell^1$ to $\R$. For the second part, we use the Cauchy residual formula, and get that
    \[
    (T + H)_- - T_- 
    = \dfrac{-1}{2 \ri \pi} \oint_{\sC} \left( \dfrac{z}{z - (T+ H)} - \dfrac{z}{z - T} \right) \rd z = \dfrac{-1}{2 \ri \pi} \oint_{\sC} \left( \dfrac{1}{z - (T+ H)} H \dfrac{1}{z - T} \right) z \rd z,
    \]
    where we used the resolvent formula in the last equality. In particular, for all $\bh \in \ell^1$ with $\| \bh \|_{\ell^1} \le \eta$ and all $z \in \sC$, we have, using Lemma~\ref{lem:a_and_A},
    \begin{equation} \label{eq:typical_bound}
        \left\| \dfrac{1}{z - (T+ H)} H \dfrac{1}{z - T} \right\|_{\fS_1} 
        \le \left\| \dfrac{1}{z - (T+ H)} \right\|_{\rm op} \left\| H \right\|_{\fS_1} \left\| \dfrac{1}{z - T} \right\|_{\rm op} \le 2 C^2  \| \bh \|_{\ell^1}.
    \end{equation}
    Integrating $z$ in the compact $\sC$ finally shows that $\cF_\bt(\bh)$ is well-defined and continuous around $\bnull$ for the $\ell^1$ norm.
    
    \medskip
    
    We can push the resolvent formula, and write that
    \[
    \frac{1}{z-(T+H)} - \frac{1}{z-T} =  \sum_{n=1}^\infty  \frac{1}{z-T} \left(H  \frac{1}{z-T} \right)^n,
    \]
    where the sum on the right is absolutely convergent in $\cB$ whenever
    \[
    \sup_{z \in \sC} \left\| \frac{1}{z-T}H \right\|_{\rm op} < 1, 
    \]
    which happens whenever $\| \bh \|_{\ell^\infty}$ is small enough, according to Lemma~\ref{lem:a_and_A}. Actually, it is also absolutely convergent in $\fS_1$ whenever $\| \bh \|_{\ell^1}$ is small enough. We deduce directly that $\bh \mapsto \cF_\bt(\bh)$ is analytic on an $\ell^1$ neighborhood of $\bnull$.
    
    \medskip
    
    Let us compute the differential and hessian of this map. We write
    \[
    \cF_\bt(\bh) = L_\bt(\bh) + \frac12 H_\bt(\bh, \bh) + R_\bt(\bh),
    \]
    with the linear form (differential) $L_\bt$ on $\ell^1(\Z)$, defined by
    \[
    L_\bt(\bh) := \mu \sum_{n \in \Z} (t_n - 2) h_n + 2 \Tr \left( \dfrac{1}{2 \ri \pi} \oint_{\sC}  \frac{1}{z-T}H\frac{1}{z-T}  z \rd z \right),
    \]
    the bilinear form (hessian) $H_\bt$ on $\ell^1(\Z) \times \ell^1(\Z)$, defined by
    \[
    H_\bt(\bh, \bk) := \mu \sum_{n \in \Z} h_n k_n + 4  \Tr \left(  \dfrac{1}{2 \ri \pi} \oint_{\sC} \frac{1}{z-T}H\frac{1}{z-T} K \frac{1}{z - T} z \rd z \right) ,
    \]
    and the rest
    \[
    R_\bt(\bh) :=  2\Tr \left( \dfrac{1}{2 \ri \pi} \oint_{\sC}  \left[\frac{1}{z-T}H \right]^3 \frac{1}{z - T} z \rd z \right) .
    \]
    
    Reasoning as in~\eqref{eq:typical_bound}, we see that $| R_\bt(\bh) | \le C \| \bh \|_{\ell^3}^3 \le C \| \bh \|_{\ell^2}^3$. Similarly, we have
    \[
    \left| H_\bt(\bh, \bk) \right| \le C \| H \|_{\fS_2} \| K \|_{\fS_2} \le C' \| \bh \|_{\ell^2} \| \bk \|_{\ell^2},
    \]
    so the bilinear form $H_\bt$ can be extended continuously on $\ell^2(\Z)$. 
    
    \medskip
    
    To end the proof of Lemma~\ref{lem:F_is_smooth}, it remains to simplify the expressions of $L_\bt$ and $H_\bt$. We use the following result.
    
    \begin{lemma}
        We have
        \begin{equation} \label{eq:simplication_L}
            \Tr \left(  \dfrac{1}{2 \ri \pi} \oint_{\sC} \frac{1}{z - T} H \dfrac{1}{z - T}  z \rd z  \right) =
            \Tr \left( \dfrac{1}{2 \ri \pi} \oint_{\sC} \frac{1}{z - T} H \rd z  \right)  = \Tr(\Gamma_\bt H)
        \end{equation}
        and
        \begin{equation} \label{eq:simplication_H}
            4 \Tr \left( \dfrac{1}{2 \ri \pi} \oint_{\sC}   \frac{1}{z - T} H \dfrac{1}{z - T} K \dfrac{1}{z - T}z \rd z \right)    =  2 \Tr \left( \dfrac{1}{2 \ri \pi} \oint_{\sC} \dfrac{1}{z - T}H \dfrac{1}{z - T} K \rd z \right) .
        \end{equation}
    \end{lemma}
    
    \begin{proof}
        First, writing that $z = (z-T) + T$, we get
        \[
        \dfrac{1}{2 \ri \pi}\oint_{\sC} \dfrac{z \rd z}{(z - T)^2} =
        \dfrac{1}{2 \ri \pi}\oint_{\sC} \dfrac{ \rd z }{z - T} + \dfrac{T}{2 \ri \pi}\oint_{\sC} \dfrac{ \rd z }{(z - T)^2},
        \]
        and the second term vanishes by the Cauchy residual formula. We recognize the spectral projector 
        \[
        \Gamma_\bt := \1(T < 0) = \dfrac{1}{2 \ri \pi} \oint_\sC \dfrac{\rd z}{z - T},
        \]
        in the first term. This and the cyclicity of the trace gives~\eqref{eq:simplication_L}. We now differentiate this equality with respect to $T$, in the direction $K$. We get
        \begin{align*}
            & \dfrac{1}{2 \ri \pi} \oint_{\sC} \Tr \left( \frac{1}{z - T} K \dfrac{1}{z - T} H \dfrac{1}{z - T} +  \frac{1}{z - T} H \dfrac{1}{z - T} K \dfrac{1}{z - T} \right) z \rd z \\
            & \quad  = 
            \dfrac{1}{2 \ri \pi} \oint_\sC \Tr \left( \dfrac{1}{z - T} K \dfrac{1}{z - T} H \right) \rd z.
        \end{align*}
        Using again the cyclicity of the trace gives~\eqref{eq:simplication_H}.
    \end{proof}

    \section{Proofs of the Lemmas }
    \label{sec:proofs}
    
    In this section, we provide the proofs of the Lemmas appearing in Section~\ref{ssec:strategy_proof}.
    
    \subsection{Proof of Lemma~\ref{lem:Talphas}}
    \label{ssec:proof:Talphas}
    
    We first prove Lemma~\ref{lem:Talphas}, which compares $T^+$ and $\widetilde{T}_{\alpha, s}^+$. The fact that $\Theta_{\alpha, s} T^+ = \widetilde{T}_{\alpha, s}^+ \Theta_{\alpha, s}$ is a simple computation. Taking the adjoint gives the second equality of~\eqref{eq:key_equality_T_Theta}.
    
    Let us now prove that $(z - \widetilde{T}_{\alpha, s})$ is invertible for $\alpha$ small enough. The operator $T^+ - \widetilde{T}_{\alpha, s}^+$ satisfies
    \[
    \left( T^+ - \widetilde{T}_{\alpha, s}^+ \right)_{n, n+1} = \begin{cases} 
        t_n^+(1 - \re^{- \alpha}) & \ \text{if} \ n < s \\
        0 & \ \text{if} \ n \ge s
    \end{cases}, \quad 
    \left( T^+ - \widetilde{T}_{\alpha, s}^+ \right)_{n+1, n} = \begin{cases} 
        t_n^+(1 - \re^{\alpha}) & \ \text{if} \ n < s \\
        0 & \ \text{if} \ n \ge s
    \end{cases},
    \]
    and $\left( T^+ - \widetilde{T}_{\alpha, s}^+ \right)_{i,j} = 0$ if $| i -j | \neq 1$. Reasoning as in Lemma~\ref{lem:a_and_A}, we deduce that
    \begin{equation} \label{eq:norm_T-Talphas}
        \| T^+ - \widetilde{T}_{\alpha, s}^+ \|_{\rm op} \le 2 \max_{n \in \Z} | t_n^+ |\cdot \max \{ |1 - \re^{-\alpha}|, |1 - \re^{\alpha}|\} = 2 (W + \delta) (\re^{\alpha} - 1).
    \end{equation}
    This bound is independent of $s \in \Z$, and goes to $0$ as $\alpha \to 0$. Since $(z - T^+)$ is invertible for all $z \in \sC$, we deduce that for $\alpha$ small enough, $(z - \widetilde{T}_{\alpha, s}^+)$ is invertible with bounded inverse, and satisfies $\| (z - \widetilde{T}_{\alpha, s}^+)^{-1} \|_{\rm op} \le C$ for a constant $C$ independent of $z \in \sC$.
    
    Finally, from the equality $\Theta_{\alpha, s} T^+ = \widetilde{T}_{\alpha, s}^+ \Theta_{\alpha, s}$, we get $\Theta_{\alpha, s} (z - T^+) = (z - \widetilde{T}_{\alpha, s}^+) \Theta_{\alpha, s}$, which gives, as wanted
    \[
    \Theta_{\alpha, s}\dfrac{1}{z - T^+} = \dfrac{1}{z - \widetilde{T}_{\alpha,s}^+}\Theta_{\alpha, s}  , 
    \quad \text{and} \quad
    \dfrac{1}{z - T^+} \Theta_{\alpha, s} = \Theta_{\alpha, s} \dfrac{1}{z - (\widetilde{T}_{\alpha,s}^+)^*} .
    \]

    \subsection{Proof of Lemma~\ref{lem:sL_invertible}}
    \label{ssec:proof:sL_invertible}
    
    We now prove Lemma~\ref{lem:sL_invertible}. The first and most important step is to prove the following Proposition, whose proof is postponed to Section~\ref{ssec:proof_coercivity}
    \begin{proposition}\label{prop-coercive}
        For the dimerized configuration $\bt^+$, the hessian $H_{\bt^+}$ is bounded on $\ell^2(\Z) \times \ell^2(\Z)$ and coercive.
    \end{proposition}
    
    Using this result, we can prove that $\sL$ is a symmetric bounded invertible operator on $\ell^2(\Z)$. Recall that $\sL$ is defined in~\eqref{eq:def:sL} by
    \[
    (\sL\bu)_n = \mu u_n +  4 \left( \frac{1}{2\ri\pi}\oint_{\sC} \left( \frac{1}{z-T^+}U\frac{1}{z - T^+} \right) \rd z \right)_{n, n+1}.
    \]
    As we already noticed in~\eqref{eq:vLu}, we have 
    \[
    \langle \bv, \sL \bw \rangle_{\ell^2} = \langle \sL \bv,  \bw \rangle_{\ell^2} = H_{\bt^+}(\bv, \bw).
    \]
    This equality is first valid for $\bv, \bw$ compactly supported, but can be extended for $\bv, \bw \in \ell^2(\Z)$ by continuity of $H_{\bt^+}$. This already proves that $\sL$ is a symmetric bounded operator on $\ell^2(\Z)$. In addition, the coercivity of $H_{\bt^+}$ shows that $\sL$ is invertible with bounded inverse (Lax-Milgram theorem).
    
    \medskip
    
    We now focus on the map $\widetilde{\sL}_{\alpha, s}$ defined in~\eqref{eq:def:L_alpha_s}. We claim that $\| \widetilde{\sL}_{\alpha, s} - \sL \|_{\rm op}$ goes to $0$ as $\alpha \to 0$. This will eventually prove that $\widetilde{\sL}_{\alpha, s}$ is also invertible with bounded inverse.
    
    We have, for $\bv, \bw \in \ell^2(\Z)$ , 
    \begin{align*}
        \langle \bv, (\sL - \widetilde{\sL}_{\alpha, s}) \bu \rangle_{\ell^2} &
        = 
        2 \Tr \left( \frac{1}{2 \ri \pi} \oint_{\sC} \left[ \dfrac{1}{z - T^+} U \frac{1}{z - T^+} V - \dfrac{1}{z - \widetilde{T}^+_{\alpha, s}} U \frac{1}{z - (\widetilde{T}^+_{\alpha, s})^*} V  \right]  \rd z \right).
    \end{align*}
    We have
    \begin{align*}
        & \dfrac{1}{z - T^+} U \frac{1}{z - T^+} - \dfrac{1}{z - \widetilde{T}^+_{\alpha, s}} U \frac{1}{z - (\widetilde{T}^+_{\alpha, s})^*}
        = \\
        & \quad = \left( \dfrac{1}{z - T^+} - \dfrac{1}{z - \widetilde{T}^+_{\alpha, s}} \right) U \frac{1}{z - T^+} 
        + 
        \dfrac{1}{z - \widetilde{T}^+_{\alpha, s}} U \left( \frac{1}{z - T^+} - \frac{1}{z - (\widetilde{T}^+_{\alpha, s})^*} \right) \\
        & \quad =  
        \dfrac{1}{z - T^+} (\widetilde{T}^+_{\alpha, s} - T^+) \dfrac{1}{z - \widetilde{T}^+_{\alpha, s}}  U \frac{1}{z - T^+} 
        + \dfrac{1}{z - \widetilde{T}^+_{\alpha, s}} U  \frac{1}{z - T^+} ((\widetilde{T}^+_{\alpha, s})^* - T^+) \frac{1}{z - (\widetilde{T}^+_{\alpha, s})^*}.
    \end{align*}
    We then use estimates of the form ($R$ stands for resolvent)
    \[
    \Tr (R_1 (T - \widetilde{T}) R_2 U R_3 V) \le \left\| R_1 (T - \widetilde{T}) R_2 U R_3 V \right\|_{\fS_1} \le \| R_1 \|_{\rm op} \| R_2 \|_{\rm op}  \| R_3 \|_{\rm op} \| T - \widetilde{T}  \|_{\rm op} \| U \|_{\fS_2} \| V \|_{\fS_2}.
    \]
    We deduce that there is $C \ge 0$ so that, for $\alpha$ small enough,
    \[
    \left| \langle \bv, (\sL - \widetilde{\sL}_{\alpha, s}) \bu \rangle_{\ell^2} \right| \le C \| \widetilde{T}^+_{\alpha, s} - T^+ \|_{\rm op} \| U \|_{\fS_2} \| V \|_{\fS_2} \le 2 C \| \widetilde{T}^+_{\alpha, s} - T^+ \|_{\rm op} \| \bu \|_{\ell^2} \| \bv \|_{\ell^2},
    \]
    where we used Lemma~\ref{lem:a_and_A} in the last inequality. We proved in Lemma~\ref{eq:norm_T-Talphas} that $\| \widetilde{T}^+_{\alpha, s} - T^+ \|_{\rm op} \to 0$ as $\alpha \to 0$. Together with the fact that $\sL$ is invertible, we deduce that there is $\alpha_* > 0$ and $C \ge 0$ so that, for all $0 \le \alpha < \alpha_*$ and all $s \in \Z$, the operator $\widetilde{\sL}_{\alpha, s}$ is invertible with $\| ( \widetilde{\sL}_{\alpha, s} )^{-1} \|_{\rm op} \le C$. This concludes the proof of Lemma~\ref{lem:sL_invertible}

    \subsection{Proof of Lemma~\ref{lem:Q}}
    \label{ssec:proof:Q}
    
    Finally, we focus on the map $\widetilde{Q}_{\alpha, s, U}$ defined in~\eqref{eq:def:Q}.
    First, using that $\sum_n (A)_{n,n+1}^2 \le \| A \|_{\fS_2}^2$ and estimates of the form
    \[
    \| R_1 (\Theta V) R_2 (W \Theta) R_3 \|_{\fS_2} \le \| R_1 \|_{\rm op} \| R_2 \|_{\rm op} \| R_3 \|_{\rm op} \| \Theta V \|_{\fS_4} \| W \Theta \|_{\fS_4},
    \]
    we get 
    \[
    \left\| \widetilde{Q}_{\alpha, s, U}(\bv, \bw) \right\|_{\ell^2(\Z)}^2 \le C \| \Theta_{\alpha, s} V \|_{\fS_4} \| W \Theta_{\alpha, s}\|_{\fS_4}.
    \]
    It remains to bound $\| \Theta_{\alpha, s} U \|_{\fS_4}$ by $\|\theta_{\alpha, s} \bu \|_{\ell^4}$. To do so, we follow the steps of Lemma~\ref{lem:a_and_A}, and prove that for all $1 \le p < \infty$ and all $\bu \in \ell^p(\Z)$, we have $\Theta_{\alpha, s} U$ in $\fS_p$ (in $\cB$ if $p = \infty$), and 
    \begin{equation} \label{eq:ThetaU_versus_thetau}
        \| \Theta_{\alpha, s} U \|_{\fS_p} \le C_p \|\theta_{\alpha, s} \bu \|_{\ell^p}
    \end{equation}
    for a constant $C_p$ independent of $\bu$ (and $\| \Theta_{\alpha, s} U \|_{\rm op} \le C_\infty \|\theta_{\alpha, s} \bu \|_{\ell^\infty}$ for $p = \infty$). We use below the fact that
    \begin{equation} \label{eq:theta_n+1_theta_n}
        \theta_{\alpha, s}(n) \le \theta_{\alpha, s}(n+1) \le \re^\alpha \theta_{\alpha, s}(n) .
    \end{equation}
    
    \medskip
    
    First, for $p = \infty$, we have, for $\psi \in \ell^2(\Z)$, 
    \begin{align*}
        \|  \Theta_{\alpha, s} U \psi \|_{\ell^2}^2 & = \sum_{n \in \Z} \theta_{\alpha, s}^2(n) | u_{n-1} \psi_{n-1} + u_n \psi_{n+1} |^2
        \le 2 \sum_{n \in \Z} \theta_{\alpha, s}^2(n) | u_{n-1} |^2 | \psi_{n-1} |^2  + \theta_{\alpha, s}^2(n) | u_n |^2 | \psi_{n+1} |^2 \\
        & \le 2 \re^\alpha \| \theta_{\alpha, s} \bu \|_{\ell^\infty}^2 \| \psi \|_{\ell^2}^2 + 2 \| \theta_{\alpha, s} \bu \|_{\ell^\infty}^2 \| \psi \|_{\ell^2}^2.
    \end{align*}
    We used that $|a+b|^2 \le 2 | a |^2 + 2 | b |^2$ for the first inequality, and~\eqref{eq:theta_n+1_theta_n} for the second. This proves the bound
    \[
    \|  \Theta_{\alpha, s} U \|_{\rm op}^2 \le \left( 2 \re^\alpha + 2 \right) \| \theta_{\alpha, s} \bu \|_{\ell^\infty}^2.
    \]
    
    \medskip
    
    In the case $p = 1$, we have by duality
    \begin{align*}
        \|  \Theta_{\alpha, s} U \|_{\fS_1} & = \sup_{K \in \cB \atop \| K \|_{\rm op} = 1} | \Tr(\Theta_{\alpha, s} U K) | = \sup_{K \in \cB \atop \| K \|_{\rm op} = 1} | \sum_{n \in \Z} u_n \theta_{\alpha, s}(n) K_{n+1, n} + u_n \theta_{\alpha, s}(n+1) K_{n, n+1} | \\
        & \le \sum_{n \in \Z} | u_n \theta_{\alpha, s}(n) | +  \sum_{n \in \Z} | u_n \theta_{\alpha, s}(n+1) | \le \| \theta_{\alpha, s} \bu \|_{\ell^1} + \re^\alpha \| \theta_{\alpha, s} \bu \|_{\ell^1} .
    \end{align*}
    Here, we used that both $| K_{n, n+1} |$ and $|K_{n+1, n}|$ are smaller than 1, and~\eqref{eq:theta_n+1_theta_n} for the last inequality.
    
    We conclude that~\eqref{eq:ThetaU_versus_thetau} holds for all $1 \le p \le \infty$ using Riesz-Thorin interpolation.
    
    
    \subsection{Coercivity of the Hessian at the dimerized configuration}
    \label{ssec:proof_coercivity}

    In this section, we prove Proposition \ref{prop-coercive}. Recall that
    \[
    H_{\bt^+}(\bh, \bh) = \mu \Vert\bh\Vert^2+ 2  \Tr \left(  \dfrac{1}{ 2 \ri \pi} \oint_{\sC}\frac{1}{z - T^+} H \frac{1}{z - T^+} H \rd z \right) 
    \]
    We already proved that $H_{\bt^+}$ is a bounded quadratic form on $\ell^2(\Z)$. We now prove
    that, for the dimerized configuration $\bt^+$, the Hessian $H_{\bt^+}$ is a coercive bilinear map on $\ell^2(\Z)$, namely that there is $C > 0$ so that, for all $\bh \in \ell^2(\Z)$, we have
    \[
    H_{\bt^+} (\bh, \bh) \ge C \| \bh \|^2_{\ell^2}.
    \]
    
    By density, it is enough to prove the result for all compactly supported sequences $\bh$. Assume that $\bh$ is such a sequence, so that $h_{n} = 0$ for all $| n | \ge S$. First, we claim that there is $C > 0$ so that, for all $L$ large enough, we have 
    \begin{equation} \label{eq:HL_coercive}
        H^{(L)}_{\bt^+_L}(\bh, \bh) \ge C \| \bh \|_{\ell^2}^2,
    \end{equation}
    where $H^{(L)}_{\bt^+_L}(\bh, \bh)$ is the hessian of the SSH model for the closed $L = 2N$ chain, defined by
    \[
    H^{(L)}_{\bt^+_L}(\bh, \bh) = \mu \Vert\bh\Vert^2_{\ell^2}+2 \Tr_L \left(  \dfrac{1}{ 2 \ri \pi} \oint_{\sC}  \frac{1}{z - T^+_L} H \frac{1}{z - T^+_L} H \rd z \right).
    \]
    Here, $\Tr_L$ is the trace for $L \times L$ hermitian matrices, $\bt^+_L$ is the dimerized ground state of the closed $L$-chain ($L$ even), of the dimerized form (see~\cite{KenLie04})
    \begin{equation} \label{eq:dimerized_L}
        \bt_L^+ = W_L + (-1)^n \delta_L,
    \end{equation}
    and $T^+_L$ is the associated $L \times L$ hermitian matrix. It was proved in~\cite{GarSer12, GonKouSer23} that $W_L \to W$ and $\delta_L \to \delta$ as $L \to \infty$.
    
    \medskip
    
    To prove~\eqref{eq:HL_coercive}, as in~\cite{KenLie04}, we use the convexity of the function $f : [0, 1] \to \R$ defined by
    $$
    [0,1] \ni x \mapsto - \sqrt{x} - \frac{1}{8}x^2.
    $$ 
    As a consequence, the map $A \mapsto \Tr (f(A))$ is convex on the set of hermitian matrices with spectrum in $[0, 1]$. This implies that, with $A = \frac{1}{\| T \|_{\rm op}} T^2$,
    \begin{equation*}
        - \Tr_L(\sqrt{T^2}) \ge - \Tr_L(\sqrt{\bra T^2 \ket}) + \frac{1}{8} \dfrac{1}{\| T \|_{\rm op}^{3}} \Tr_L[T^4 - \bra T^2 \ket^2],
    \end{equation*}
    where $\bra A \ket$ is the average of $A$ over all translations, namely
    \[
    \bra A \ket = \frac{1}{L} \sum_{k=0}^{L-1} \Theta_1^k A \Theta_1^{-k}, 
    \qquad 
    \Theta_1 = \begin{pmatrix}
        0 & 1 & 0 & \cdots & 0 \\
        0 & 0 & 1 & 0 & \cdots \\
        \vdots & \vdots & \vdots & \ddots &\vdots \\
        0 & 0 & \cdots & 0 & 1 \\
        1 & 0 & \cdots & 0 & 0
    \end{pmatrix}.
    \]
    We deduce that
    \begin{equation} \label{eq:ineq_EL}
        \cE^{(L)}(\bt) \ge \frac{\mu}{2}\sum_{n\in\mathbb{Z}/L\mathbb{Z}}^{}(t_n - 1)^2 - \Tr_L(\sqrt{\bra T^2 \ket}) + \frac{1}{8} \dfrac{1}{\| T \|_{\rm op}^{3}} \Tr_L[T^4 - \bra T^2 \ket^2].
    \end{equation}
    
    Since the operator $T_L^{\pm}$ always corresponds to a $2$-periodic configuration, it holds
    \[
    \Theta_1 \left( T_L^{\pm} \right)^2 \Theta_1^* = \left( T_L^{\pm} \right)^2
    \quad \text{and, in particular} \quad 
    \left( T_L^{\pm} \right)^2 = \left\bra  \left( T_L^{\pm} \right)^2 \right\ket.
    \]
    We deduce that for $\bt = \bt^+_L$, the third term of~\eqref{eq:ineq_EL} vanishes. Actually, the proof of Kennedy and Lieb~\cite{KenLie04} shows that $\bt_L^\pm$ also minimizes the quantity
    \[
       \bt \mapsto \frac{\mu}{2}\sum_{n\in\mathbb{Z}/L\mathbb{Z}}^{}(t_n - 1)^2 - \Tr_L(\sqrt{\bra T^2 \ket}).
    \]
    We deduce that
    \[
    \cE^{(L)}(\bt) - \cE^{(L)}(\bt^+_L) \ge  \frac{1}{8} \dfrac{1}{\| T \|_{\rm op}^{3}} \Tr_L[T^4 - \bra T^2 \ket^2].
    \]
    We apply this inequality for $\bt = \bt^+_L + s \bh$ and get that
    \[
    \cE^{(L)}(\bt^+_L + s \bh) - \cE^{(L)}(\bt^+_L) \ge \frac{1}{8} \dfrac{1}{\| T^+_L + s H \|_{\rm op}^3} \Tr_L \left[ (T^+_L + sH)^4 - \bra (T^+_L + sH)^2 \ket ^2 \right]
    \]
    For the denominator, we use the fact that, for $s$ small enough, we have $\| T^+_L + s H \|_{\rm op} \le 2 \| T^+_L \|_{\rm op}$. For the numerator, expanding the expression and using that $ (T_L^+)^2  = \langle (T_L^+)^2 \rangle$, so that $\Tr_L( (T_L^+)^3H) = \Tr_L \left( \langle (T_L^+)^2 \rangle \langle T_L^+ H \rangle \right)$ and $\Tr_L(\langle (T_L^+)^2 \rangle \langle H^2 \rangle) = \Tr_L((T_L^+)^2 H^2)$ (the orders $O(1)$ and $O(s)$ vanish), we obtain
    \begin{align*}
        \cE^{(L)}(\bt^+_L + s \bh) - \cE^{(L)}(\bt^+_L)
        & \ge \frac{s^2}{16} \dfrac{1}{\| T^+_L \|_{\rm op}^3} \Tr_L \Big[ \left( T_L^+ H + H T_L^+ \right)^2  -   \langle T_L^+ H + H T_L^+\rangle^2  \Big] + o(s^2) \\
        & = \frac{s^2}{16} \dfrac{1}{\| T^+_L \|_{\rm op}^3} \left( \left\| T_L^+ H + H T_L^+ \right\|_{\fS_2}^2  -  \left\|  \langle T_L^+ H + H T_L^+\rangle \right\|_{\fS_2}^2  \right) + o(s^2)
    \end{align*}

    The previous computation is valid for all $\bh$. We now use the fact that $\bh$ is compactly supported in $[-S, S]$, and that $L \gg S$. This allows to prove that the last term is small. More specifically, we have the following.
    \begin{lemma}
        For all $S \in \N$ all $L \gg S$, all $\bt \in \C^L$ and all $\bh \in \C^L$ compactly supported in $[-S, S]$, we have
        \[
        \left\| \bra TH \ket \right\|_{\fS_2}^2 \le \frac{6S}{L} \| \bt \|_{\ell^\infty}^2 \| \bh \|_{\ell^2}^2.
        \]
    \end{lemma}
    
    \begin{proof}
        Set $A = TH$, so that 
        \[
        A_{n,n} = t_{n} h_{n} + h_{n-1} t_{n-1}, \quad
        A_{n, n+2} = t_{n+1} h_n, \quad
        A_{n, n-2} = t_{n-2} h_{n-1}.
        \]
        The matrix $\bra A \ket$ is of the form
        \[
        \bra A \ket_{n,n} = 2 a_0, \quad \bra A \ket_{n,n+2} = a_1, \quad \bra A \ket_{n, n-2} = a_{-1}, \qquad \text{with} \quad 
        a_m := \frac{1}{L} \sum_{k=0}^{L-1} t_{k+m} h_{k}.
        \]
        Using that $h$ is compactly supported, by Cauchy-Schwarz we get
        \[
        | a_m |^2 = \frac{1}{L^2} \left( \sum_{k=-S}^S t_{k+m} h_k \right)^2 
        \le \frac{1}{L^2} \| \bt \|_{\ell^\infty}^2 \left( \sum_{k=-S}^S h_k \right)^2 
        \le \frac{1}{L^2} \| \bt \|_{\ell^\infty}^2 \| \bh \|_{\ell^2}^2 S.
        \]
        We obtain
        \[
        \| \bra A \ket \|_{\fS_2}^2 = \sum_{n,m} \bra A \ket_{n,m}^2 = 
        L (2 a_0)^2 + L a_1^2 + L a_{-1}^2
        \le 6 L \frac{1}{L^2} \| \bt \|_{\ell^\infty}^2 \| \bh \|_{\ell^2}^2 S.
        \]
    \end{proof}
    This proves that 
    \begin{align*}
        \cE^{(L)}(\bt^+_L + s \bh) - \cE^{(L)}(\bt^+_L)
        & = \frac{s^2}{16} \dfrac{1}{\| T^+_L \|_{\rm op}^3} \left( \left\| T_L^+ H + H T_L^+ \right\|_{\fS_2}^2  -  \frac{12 S}{L} \| \bt_L^+ \|_{\ell^\infty}^2 \| \bh \|_{\ell^2}^2 \right) + o(s^2).
    \end{align*}
    Finally, we bound from below the remaining $\left\| T_L^+ H + H T_L^+ \right\|_{\fS_2}^2$. A computation shows that the matrix $A := T H + HT$ satisfies
    \[
    A_{n,n} = 2 (t_n h_n + t_{n-1} h_{n-1}) , \quad A_{n, n+2} = t_n h_{n+1} + h_n t_{n+1}, \quad A_{n, n-2} = t_{n-1} h_{n-2} + h_{n-1} t_{n-2},
    \]
    and $A_{i,j} = 0$ otherwise. Squaring all terms and summing gives
    \begin{align*}
        \left\| T_L^+ H + H T_L^+ \right\|_{\fS_2}^2 & = \sum_{n} 4 (t_n h_n + t_{n-1} h_{n-1})^2 + ( t_n h_{n+1} + h_n t_{n+1} )^2 + (t_{n-1} h_{n-2} + h_{n-1} t_{n-2})^2.
    \end{align*}
    Expanding, relabelling all sums, and using that $\bt_L^+$ is dimerized, of the form~\eqref{eq:dimerized_L}, we obtain 
    \[
    \left\| T_L^+ H + H T_L^+ \right\|_{\fS_2}^2 = 2 \sum_{n \in \Z} \left\bra \begin{pmatrix} h_n \\ h_{n+1} \end{pmatrix}, Q_n \begin{pmatrix} h_n \\ h_{n+1} \end{pmatrix} \right\ket, 
    \quad \text{with} \quad 
    Q_n = \begin{pmatrix}
        2 t_n^2  +   t_{n+1}^2 & 3 t_n t_{n+1} \\ 3 t_n t_{n+1} &  t_n^2  +  2 t_{n+1}^2
    \end{pmatrix}.
    \]
    We have
    \[
    \Tr(Q_n) = 3 (t_n^2 + t_{n+1}^2) = 6(W_L^2 + \delta_L^2)
    \]
    and
    \[
    \det Q_n = ( 2 t_n^2  +   t_{n+1}^2 ) (t_n^2  +  2 t_{n+1}^2) - 9 t_n^2 t_{n+1}^2 = 2 t_n^4 + 2 t_{n+1}^4 - 4 t_n^2 t_{n+1}^2 = 2 (t_n^2 - t_{n+1}^2)^2
    = 32 \, W_L^2 \delta_L^2.
    \]
    Since $\delta_L \to \delta > 0$ and $W_L \to W$ for $L$ large enough, there is a constant $C \ge 0$ such that $Q_n \ge C > 0$ for a constant $C$ independent of $n$ and $L$ large enough. So 
    \[
    \left\| T_L^+ H + H T_L^+ \right\|_{\fS_2}^2  \ge 2 C \| \bh \|_{\ell^2}^2.
    \]
    Altogether, we proved that for $L$ large enough, we have
    \begin{align*}
        \cE^{(L)}(\bt^+_L + s \bh) - \cE^{(L)}(\bt^+_L)
        & \ge \frac{s^2}{16} \dfrac{1}{\| T^+_L \|_{\rm op}^3} \left( 2 C  -  \frac{12 S}{L} \| \bt_L^+ \|_{\ell^\infty}^2 \right) \| \bh \|_{\ell^2}^2  + o(s^2) \\
        & \ge \widetilde{C} s^2  \| \bh \|_{\ell^2}^2  + o(s^2),
    \end{align*}
    where $\widetilde{C}$ is independent of $L$, for $L$ large enough ($L \ge L_0$, where $L_0$ depends on the support $S$ of $\bh$). This proves the lower bound~\eqref{eq:HL_coercive} for $H_{\bt^+_L}^{(L)}$. 
    
    To conclude the proof, we note that 
    \begin{align*}
        \left| \Tr_L \left( \frac{1}{z - T^+} H \frac{1}{z - T^+} H \right) 
        - \Tr_L \left( \frac{1}{z - T^+_L} H \frac{1}{z - T^+_L} H \right)  \right| \le C \| H \|_{\fS_2}^2 \| T^+ - T^+_L \|_{\rm op, L}.
    \end{align*}
    Since $W_L \to W$ and $\delta_L \to \delta$, we have $\| T^+ - T^+_L \|_{\rm op, L} \to 0$ as $L = 2N$ (even) goes to infinity. So for $L$ large enough, we have
    \[
    \left| H_{\bt^+}(\bh, \bh) - H_{\bt_L^+}^{(L)}(\bh, \bh) \right| \le \frac{C}{2} \| \bh \|_{\ell^2}^2,
    \] 
    where $C$ is the bound in~\eqref{eq:HL_coercive}. This proves $H_{\bt^+}(\bh, \bh) \ge \frac{C}{2} \| \bh \|_{\ell^2}^2$, where the constant $C$ is independent of $\bh$. We proved the bound for $\bh$ compactly supported, but by density, it can be extended for all $\bh \in \ell^2(\Z)$, hence the coercivity of $H_{\bt^+}$.

    \bibliographystyle{plain}
    \bibliography{biblio}
    
\end{document}